\documentclass[12pt, draftclsnofoot, onecolumn]{IEEEtran}
 \usepackage{amsmath,amssymb}
 \usepackage{subfigure}
 \usepackage{graphicx,graphics,color,psfrag}
 \usepackage{cite,balance}
 \usepackage{caption}
 \allowdisplaybreaks
 \usepackage{algorithm}
 \usepackage{accents}
 \usepackage{amsthm}
 \usepackage{bm}
 \usepackage{algorithmic}
 \usepackage[english]{babel}
 \usepackage{multirow}
 \usepackage{enumerate}
 \usepackage{cases}
 \usepackage{stfloats}
 \usepackage{dsfont}
 \usepackage{color,soul}
 \usepackage{amsfonts}
 \usepackage{cite,graphicx,amsmath,amssymb}
 \usepackage{subfigure}
 \usepackage{fancyhdr}
 \usepackage{hhline}
 \usepackage{graphicx,graphics}
 \usepackage{array,color}
 \usepackage{amsmath}

\newtheorem{theorem}{Theorem}

\newtheorem{lemma}{Lemma}

\newtheorem{remark}{\bf Remark}

\def\E{\mathsf{E}}

\def\phi{\varphi}

\def\l{\left}
\def\r{\right}
\def\({\left(}
\def\){\right)}

\def\b0{{\mathbf{0}}}

\newcommand{\nn}{\nonumber}

\begin{document}

\addtolength{\textfloatsep}{-15pt}
%

\title{\huge Connectivity and Blockage Effects in Millimeter-Wave Air-To-Everything Networks}
\author{Kaifeng Han, Kaibin Huang and Robert W. Heath Jr. \thanks{K. Han and K. Huang are with the Dept. of EEE at The  University of  Hong Kong, Hong Kong (e-mail: \{kfhan, haungkb\}@eee.hku.hk). R. W. Heath, Jr. is with The University of Texas at Austin, Austin, TX 78712 USA (e-mail: rheath@utexas.edu). The corresponding author is R. W. Heath, Jr.}}
\maketitle

\begin{abstract}
\emph{Millimeter-wave} (mmWave) offers high data rate and bandwidth for \emph{air-to-everything} (A2X) communications including  air-to-air, air-to-ground, and air-to-tower. MmWave communication in the A2X network is sensitive to buildings blockage effects.
In this paper, we  propose an analytical framework to define and characterise the connectivity for an \emph{aerial access point} (AAP) by jointly using stochastic geometry and random shape theory. The buildings are modelled as a Boolean line-segment process with fixed height. The blocking area for an arbitrary building is derived and minimized by optimizing the altitude of AAP. A  lower bound on the connectivity probability is derived as a function of the altitude of AAP and different parameters  of users and buildings including their densities, sizes, and heights. Our study yields guidelines on practical mmWave A2X networks deployment.
\end{abstract}


\section{Introduction}
\emph{Air-to-everything} (A2X) communication  can leverage \emph{aerial access points} (AAPs) mounted on \emph{unmanned aerial vehicles} (UAVs) to provide seamless wireless connectivity to various types of users \cite{zeng2016wireless} (see Fig.~\ref{SystemModel}). \emph{Millimeter-wave} (MmWave) communication is one way to provide high data rates for aerial platforms \cite{xiao2016enabling}. Unfortunately, mmWave communication is sensitive to building blockages  \cite{andrews2017modeling}, which are widely expected in urban deployments of AAPs. In this paper, we define and characterize the connectivity for an AAP, using tools from stochastic geometry and random shape theory.


\begin{figure}[t]
\centering
\includegraphics[width=9cm]{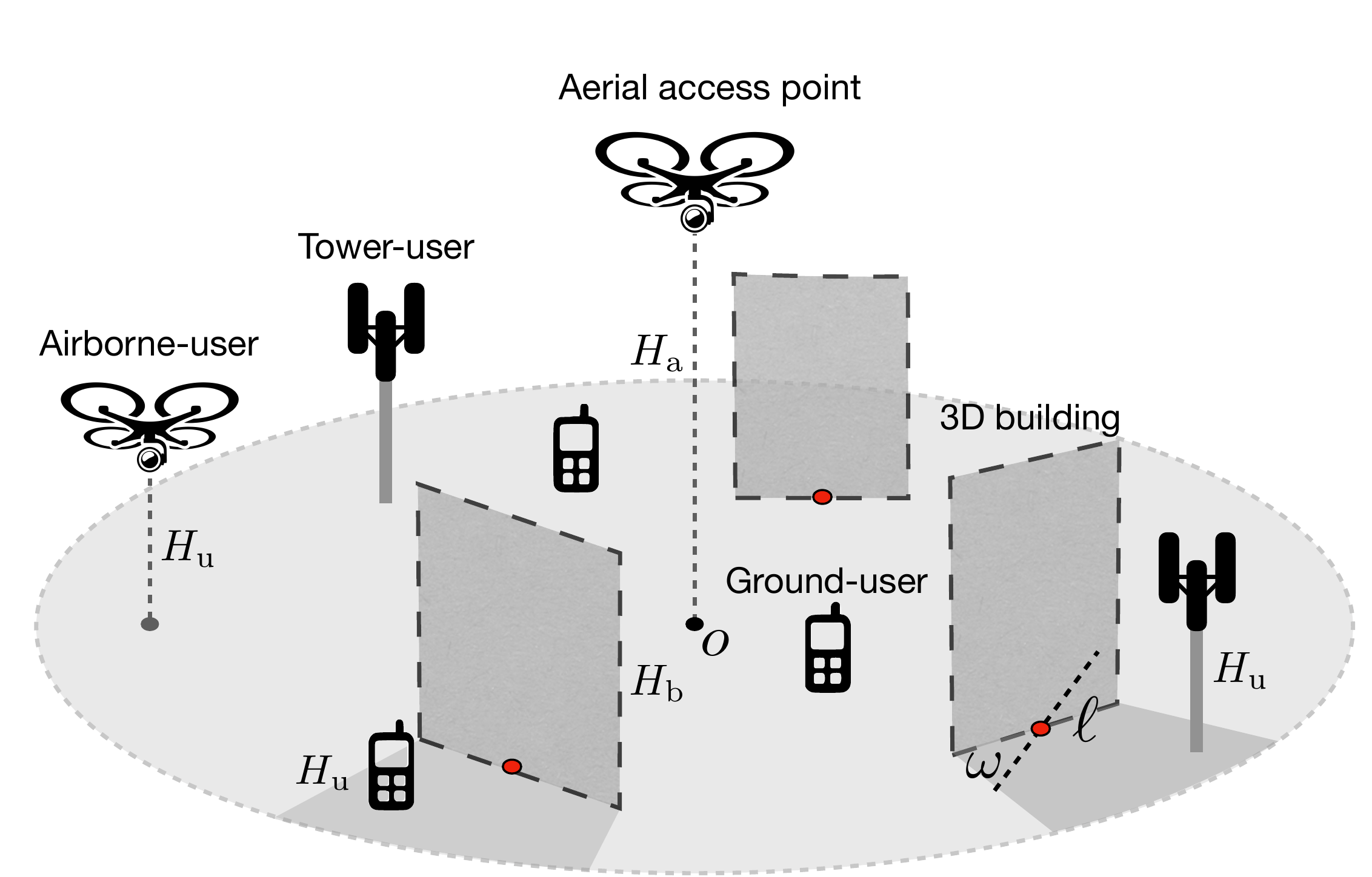}
\caption{An illustration of the A2X communication network. A typical (central) AAP provides wireless connectivity to different types of users, including AAP connects with ground-users (e.g., mobile) via \emph{air-to-ground} (A2G), tower-users (e.g., \emph{base station}, BS) via \emph{air-to-tower} (A2T), and airborne-users (e.g., UAV) via \emph{air-to-air} (A2A) communications. The altitude of AAP is denoted by $H_\textrm{a}$, and the height of users are denoted by $H_\textrm{u}$. The 3D buildings are modelled as a  Boolean line-segment process with fixed height $H_\textrm{b}$. }\label{SystemModel}
\end{figure}

Leveraging UAVs as AAPs  has been studied in recent literature \cite{al2014optimal, lyu2017placement, chetlur2017downlink, mozaffari2016unmanned, galkin2017stochastic, hou2018multiple, cuvelier2018mmwave}. A single-UAV network was proposed in \cite{al2014optimal}, where the network coverage was maximized by optimizing the UAVs' altitudes. The coverage performance can also be maximized via optimizing the placement of UAVs \cite{lyu2017placement}. In \cite{chetlur2017downlink}, the coverage probability of a finite 3D multi-UAV network was calculated via a stochastic geometric  approach. Both network coverage and the sum-rate of a hybrid A2G-D2D network were investigated in \cite{mozaffari2016unmanned}. An analytical framework that UAV uses ground-BS for wireless backhaul was proposed in \cite{galkin2017stochastic} with providing the analysis for success probability of establishing a backhaul link as well as backhaul data rate. In \cite{hou2018multiple}, the \emph{multiple-input multiple-output} (MIMO) \emph{non-orthogonal multiple access} (NOMA) techniques were used in UAV network and the outage probability and ergodic rate of network were studied based on a stochastic geometry model. In \cite{al2014optimal, mozaffari2016unmanned, galkin2017stochastic}, the blockage effects were characterized by a statistical model where the link-level \emph{line-of-sight} (LOS) probability is approximated as a simple sigmoid function. The parameters of sigmoid function are determined the by buildings' density, sizes, and heights' distribution. The model is unsuitable for mmWave A2X networks since it fails to capture the fact that multiple nearby links could be simultaneously blocked by the same building and does not consider the diversity in user types  (e.g., their  different heights). In  \cite{cuvelier2018mmwave},  a mathematical framework was proposed  for studying mmWave A2A networks, in which multiple aerial-users are equipped with antenna arrays. Blocking effects were not included since A2A scenario was assumed to be well above the blockages.

In this paper, we develop an analytical framework for characterizing the blockage effects and connectivity of a mmWave A2X network covered by a single AAP. The  3D buildings are modelled as a  Boolean line-segment process with fixed height. Given an arbitrary building, the corresponding blocking area is derived as a function of altitude of AAP, users and buildings' parameters in'cluding their density, sizes, and heights. Based on the model, the AAP coverage area is maximized (or equivalently  the blocking area minimized) by optimizing the altitude of AAP. Furthermore, both upper and lower bounds on the  blocking area and a suboptimal result of AAP's altitude are derived in closed-form. Finally, the spatial average connectivity probability of a typical A2X network is obtained, which may be maximized by optimizing the AAP's altitude.

\section{System Model and Performance Metric}
Consider a mmWave   A2X network as illustrated in Fig.~\ref{SystemModel}. In this letter, we focus on the downlink communication from a typical low-altitude AAP to users with different heights.

\subsection{Channel Model between AAP and Users}

The mmWave channel between the AAP and the different types of users is assumed to be LOS or blocked by a building. For simplicity, we assume the \emph{non-LOS} (NLOS) signals are completely blocked due to severe propagation loss from penetration and limited reflection, diffraction, or scattering \cite{andrews2017modeling}. For the LOS case, the channel is assumed to have  path-loss without  small-scale fading \cite{han2018connectivity}. We assume perfect 3D  beam alignment  between AAP and users for  maximal directivity  gain \cite{cuvelier2018mmwave}. For the path-loss model, we assume the reference distance is $1\mathrm{m}$. The AAP transmission  with power $P$ and  propagation distance $r$ is attenuated modelled as  $r^{-\alpha}$ where $\alpha$ is the path-loss exponent \cite{han2018connectivity}. Let $\sigma^2$ be the thermal noise power normalized by the transmit power $P$. The corresponding \emph{signal-to-noise ratio} (SNR)  received at user is defined as $P_{\textrm{r}} = \frac{Gr^{-\alpha}}{\sigma^2}$ where $G$ denotes the beamforming gain. We assume that the user is connected to the AAP if the  receive SNR exceeds a given threshold $\gamma$. We say that the AAP has a maximal \emph{coverage sphere} with the radius $R_{{\textrm{max}}} = \l(\frac{G}{\sigma^2 \gamma}\r)^{\frac{1}{\alpha}}$.  The 2D projection of the coverage sphere of the AAP into the plane with user's height $H_\textrm{u}$ forms a disk with the radius $\Lambda_H = \sqrt{R_{{\textrm{max}}}^2 - \l(H_\textrm{a}-H_\textrm{u}\)^2}$, called the \emph{efficient coverage disk} and denoted by $\mathcal{O}(\Lambda_H)$. The user is connected to the AAP if its 2D location is inside the  efficient coverage disk and the link between the user and the AAP is LOS. For higher users heights, i.e., larger $H_{\textrm{u}}$, the efficient coverage disk is larger. Let the center of efficient coverage disk, i.e., the 2D projection of AAP's location, be the origin denoted by $o\in \mathds{R}^2$.

\subsection{3D Building Model}
A 3D building model is adopted to characterize blockage effects where buildings are modelled as the Boolean line-segment process with the same fixed height $H_{\textrm{b}}$ for tractability \cite{gupta2018macrodiversity, li2014impact}. Adding randomness to the buildings height will be left to future work. Specifically, buildings are approximated as line segments with random length on the 2D plane. Although the buildings have polygon shapes in practice, we are interested in their 1D intersections with the communication links and thus assuming their
shape as lines is a reasonable approximation. The center locations of the line-segments are modelled as a homogeneous PPP $\Phi = \{x\}$ on $\mathds{R}^2$ plane with density $\lambda_{\textrm{b}}$. The lengths $\{\ell\}$ and orientations $\{\omega\}$ of blockage line-segments are independent identically distributed random variables. Let $f_{L}(\ell)$ be the distribution of $\ell$ and let $f_{\Theta}(\omega)$ be that of $\omega$. The Boolean line-segments model can be extended to other models as discussed in  Remark \ref{Remark:ExtendsionBuildingWithShape}.

\subsection{Connectivity and Performance Metrics}

We assume that all the users can be simultaneously connected to the A2X network if they are in AAP's coverage sphere and the link between user and AAP is LOS without being blocked by any building. Consider an arbitrary building whose 2D line-segment center is located at $x \in \mathds{R}^2$. The building results in a \emph{blockaging area} $\mathcal{S}_{\textrm{b}}(x)$ where the links between users and AAP are fully blocked by the building (see Fig.~\ref{Geometry}). To measure the network performance, we define the spatial average connectivity probability, denoted by $p_{\textrm{c}}$, as the spatial average fraction of the A2X network that is connectable at any time \cite{28andrews2011tractable}. The $p_{\textrm{c}}$ is mathematically expressed as
\begin{align}\label{Def:CoverageProb}
p_{\textrm{c}} = 1 - \E\l(\frac{\l|\bigcup_{x\in \{\Phi\cap\mathcal{O}(\Lambda_H)\}}\mathcal{S}_{\textrm{b}}(x)\r|}{\l|\mathcal{O}(\Lambda_H)\r|}\r),
\end{align}
where $\l|\mathcal{O}(\Lambda_H)\r| = \pi \Lambda_H$ denotes the size of $\mathcal{O}(\Lambda_H)$.

\section{Analysis For Network Connectivity}

\begin{figure}[t]
\centering
\includegraphics[width=9.5cm]{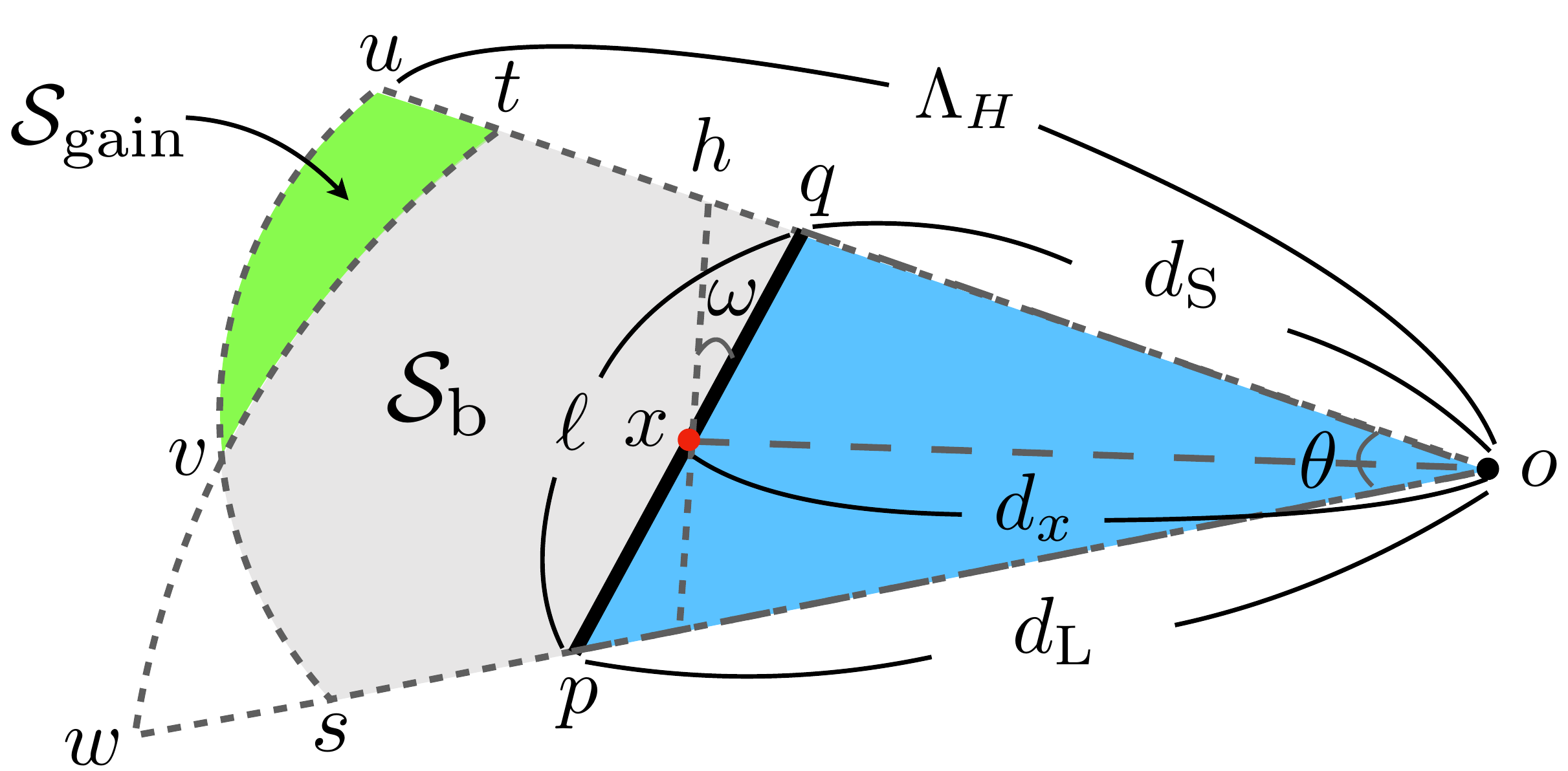}
\caption{2D projection of an arbitrary building modelled by a line-segment $pq$. Some geometrical relations are  described as follows. $d_{\textrm{S}}=\|o-q\|$, $d_{\textrm{L}}=\|o-p\|$, $\Lambda_H = \|o-u\| = \|o-s\|$, $\omega = \angle qxh$, $\theta = \angle qop$, and $\angle hxo = \pi/2$. The grey area covered by $qpsvt$ is the blocking area $\mathcal{S}_{\textrm{b}}(x)$ and the area covered by $qop$ (blue area) and $tvu$ (green area) is the coverage area. Specifically, the blue area covered by $tvu$ denotes the coverage gain $\mathcal{S}_{\textrm{gain}}$ due to the fact that higher altitude of AAP can cover more LOS area.}\label{Geometry}
\end{figure}

\subsection{Size of Blocking Area}
We begin by calculating the size of blocking area $\mathcal{S}_{\textrm{b}}(x)$ for an arbitrary building whose 2D line-segment center located at $x$. We first fix the length $\ell$ and $\omega$ of the typical building. Let $d_x$ be the distance between $x$ and $o$. Let $d_{\textrm{S}}$ be the minimal (shortest) distance between $o$ and line-segment (2D projection of building) and $d_{\textrm{L}}$ be the maximal (longest) distance (see the lines $oq$ and $op$ in Fig.~\ref{Geometry}).  If the AAP's altitude does not exceed building's height, i.e., $H_{\textrm{a}} \leq H_{\textrm{b}}$, the size of blocking area $\mathcal{S}_{\textrm{b}}(x)$ (see the gray area covered by $qpsvt$ in Fig.~\ref{Geometry}) is calculated by
$\frac{1}{2}\l[\theta \Lambda_H^2 - d_{\textrm{S}}d_{\textrm{L}}\sin\theta\r]$,
where $\theta = \arccos \l( \frac{d_x^2 - \frac{1}{4}\ell^2}{d_{\textrm{S}}d_{\textrm{L}}} \r)$ and
\begin{align}\label{Eq:dsdl}
\begin{cases}
d_{\textrm{S}} = \sqrt{\frac{1}{4}\ell^2 + d_x^2 - d_x\ell\sin\omega}, \\
d_{\textrm{L}} = \min\l[\Lambda_H, \sqrt{\frac{1}{4}\ell^2 + d_x^2 + d_x\ell\sin\omega}\r].
\end{cases}
\end{align}
If $H_{\textrm{a}} > H_{\textrm{b}}$, the blocking area $\mathcal{S}_{\textrm{b}}(x)$ could be further reduced since the AAP covers more area via LoS links due to the benefit of higher altitude. Compared with the coverage area of AAP when $H_{\textrm{a}} \leq H_{\textrm{b}}$,  we define this additional coverage area due to $H_{\textrm{a}} > H_{\textrm{b}}$ as the \emph{coverage gain}, denoted by $\mathcal{S}_{\textrm{gain}}(x)$ (see the area covered by $tvu$ (green area) in Fig.~\ref{Geometry}). Based on the basic geometric calculation, $\mathcal{S}_{\textrm{gain}}(x)$ is calculated as
\begin{align}\label{Eq:CoverageGain}
\mathcal{S}_{\textrm{gain}}(x) &= \frac{1}{2}\l[ \int\limits_{0}^{\theta^{'}}\frac{\l(d_{\textrm{S}}\cos\beta \r)^2}{\cos^2\l( \phi + \beta \r)}\mathrm{d}\phi  + \theta^{'}\Lambda_H^2 \r],
\end{align}
where $\theta^{'} = \arccos\l( \frac{d_{\textrm{S}}\cos{\beta}}{\l(1 - \frac{H_{\mathrm{b}} - H_{\mathrm{u}}}{H_{\mathrm{a}} - H_{\mathrm{u}}}\Lambda_H\r)} \r)$ and $\beta = \arctan\l( \frac{\cos{\theta} - \frac{d_{\textrm{S}}}{d_{\textrm{L}}}}{\sin{\theta}} \r)$.  Then,  $\mathcal{S}_{\textrm{b}}(x)$ is calculated as follows.
\begin{lemma}[Size of $\mathcal{S}_{\textrm{b}}(x)$]\label{Lemma:BlockingArea}\emph{The blocking area is
\begin{align}\label{Eq:BlockingArea}
\mathcal{S}_{\textrm{b}}(x) &= \frac{1}{2}\l(\theta \Lambda_H^2  - d_{\textrm{S}}d_{\textrm{L}}\sin\theta\r) - \mathbf{1}\l( H_{\textrm{a}} > H_{\textrm{b}} \r)\mathcal{S}_{\textrm{gain}}(x),
\end{align}
where $\mathbf{1}(\cdot)$ denotes the indicator function and $\mathcal{S}_{\textrm{gain}}(x)$ is given in \eqref{Eq:CoverageGain}.}
\end{lemma}
The calculations follow from geometry, the detailed proof is omitted due to limited space. To simplify the result in Lemma~\ref{Lemma:BlockingArea} and obtain more insights therein, both upper and lower bounds of $\mathcal{S}_{\textrm{gain}}(x)$ are derived. By assuming the distance between any point on building's line-segment and $o$ has the same value $d_{\textrm{L}}$ or $d_{\textrm{S}}$ given in \eqref{Eq:dsdl}, the lower and upper bounds of $\mathcal{S}_{\textrm{gain}}(x)$, denoted by $\mathcal{S}_{\textrm{gain}}^{(-)}(x)$ and $\mathcal{S}_{\textrm{gain}}^{(+)}(x)$, are derived as
\begin{align}\label{Eq:LowerBoundS_gain}
\mathcal{S}_{\textrm{gain}}^{(-)}(x) = \frac{\theta}{2}\l[\Lambda_H^2 - \l( \frac{d_{\textrm{L}}}{1 -  \frac{H_{\mathrm{b}} - H_{\mathrm{u}}}{H_{\mathrm{a}} - H_{\mathrm{u}}}} \r)^2\r]^{+},
\end{align}
and $\mathcal{S}_{\textrm{gain}}^{(+)}(x)$ is obtained by replacing $d_{\textrm{L}}$ in $\mathcal{S}_{\textrm{gain}}^{(-)}(x)$ with $d_{\textrm{S}}$. The result for the bounds of $\mathcal{S}_{\textrm{gain}}(x)$ is summarized as follows.
\begin{lemma}[Bounds of $\mathcal{S}_{\textrm{gain}}(x)$]\label{Lemma:BoundsOfCoverageGain}\emph{The coverage gain $\mathcal{S}_{\textrm{gain}}(x)$ can be upper or lower bounded as follows.
\begin{align}\label{Eq:BoundsOfGain}
\mathcal{S}_{\textrm{gain}}^{(-)}(x) \leq \mathcal{S}_{\textrm{gain}}(x) \leq \mathcal{S}_{\textrm{gain}}^{(+)}(x),
\end{align}
where $\Lambda_H$ is specified in Lemma~\ref{Lemma:BlockingArea} and $[A]^{+} = \max[0, A]$.}
\end{lemma}
The bounds for $\mathcal{S}_{\mathrm{gain}}$  can be treated as the bounds for $\mathcal{S}_{\mathrm{b}}$ via substituting \eqref{Eq:BoundsOfGain} into \eqref{Eq:BlockingArea}: $\mathcal{S}_{\textrm{b}}^{(-)}(x)\leq \mathcal{S}_{\textrm{b}}(x) \leq \mathcal{S}_{\textrm{b}}^{(+)}(x)$, where $\mathcal{S}_{\textrm{b}}^{(-)}(x) = \frac{1}{2}\l(\theta \Lambda_H^2  - d_{\textrm{S}}d_{\textrm{L}}\sin\theta\r) - \mathbf{1}\l( H_{\textrm{a}} > H_{\textrm{b}} \r)\mathcal{S}_{\textrm{gain}}^{(+)}(x)$ and $\mathcal{S}_{\textrm{b}}^{(+)}(x)$ is obtained via replacing $\mathcal{S}_{\textrm{gain}}^{(+)}(x)$ in $\mathcal{S}_{\textrm{b}}^{(-)}(x)$ with $\mathcal{S}_{\textrm{gain}}^{(-)}(x)$.
\begin{remark}[Optimal Altitude of AAP]\label{Remark:OptimalAAPAltitude}\emph{A larger AAP's altitude $H_{\mathrm{a}}$ can effectively increase the coverage (LoS) area, while shrinking the radius of effective coverage disk. We characterize this behavior by optimizing $\mathcal{S}_{\textrm{gain}}^{(-)}(x)$ to obtain a suboptimal solution for $H_{\mathrm{a}}$. When $\l( \frac{d_{\textrm{L}}}{1 -  \frac{H_{\mathrm{b}} - H_{\mathrm{u}}}{H_{\mathrm{a}} - H_{\mathrm{u}}}} \r)^2  < \Lambda_H^2$, we have
\begin{align}\label{Eq:SubOptimalGain}
H_{\mathrm{a}}^{*} &= \arg \max_{H_{\mathrm{a}}} \mathcal{S}_{\textrm{gain}}^{(-)}(x) = \l(d_{\textrm{L}}^2\l( H_{\mathrm{b}} - H_{\mathrm{u}} \r)  \r)^{\frac{1}{3}} + H_{\mathrm{b}}.
\end{align}
Substituting $H_{\mathrm{a}}^{*}$ into \eqref{Eq:BlockingArea} gives the suboptimal solution of $\mathcal{S}_\mathrm{b}(x)$.}
\end{remark}
\begin{remark}\label{Remark:ExtendsionBuildingWithShape}
\emph{Extending the current building model to any model that each building has a random size in 2D projection, such as rectangle \cite{andrews2017modeling} or disk, follows a similar analytical structure. The main difference is that the area of buildings should be included into blocking area $\mathcal{S}_{\textrm{b}}$. Also, $\mathcal{S}_{\textrm{gain}}$ needs to be recalculated based on different building's shape. For instance, if the 2D projection of a building is modelled as a disk with radius diameter $\ell$ (i.e., cylinder in 3D), the blockaging area is recalculated as $\mathcal{S}_{\mathrm{b}} = \frac{\theta}{2}\Lambda_H^2 - \l(d_x\ell + \frac{1}{8}\ell^2(\theta+\pi) \r) - \mathbf{1}\l( H_{\textrm{a}} > H_{\textrm{b}} \r)\mathcal{S}_{\textrm{gain}}(x)$, where $\mathcal{S}_{\textrm{gain}}(x)$ is lower bounded by $\mathcal{S}_{\textrm{gain}}^{(-)}(x) = \frac{\theta}{2}\l[\Lambda_H^2 - \l( \frac{ 1 }{1 -  \frac{H_{\mathrm{b}} - H_{\mathrm{u}}}{H_{\mathrm{a}} - H_{\mathrm{u}}}}\l(d_x + \frac{1}{2}\ell \r) \r)^2\r]^{+}$. }
\end{remark}
\subsection{Network Connectivity Probability}
In this section, we calculate the connectivity probability defined in \eqref{Def:CoverageProb}. Notice that the spatial correlation between different buildings exsits such as the blocking areas of multiple buildings may overlap with each other. For analytical tractablility, we ignore the spatial correlation of buildings due to overlap in the blockaging area of multiple buildings. This assumption is accurate when density of buildings is not very high, which has been validated in \cite{andrews2017modeling}. We derive a lower bound of $p_\mathrm{c}$ by jointly using Campbell's theorem, random shape theory, with the results given in Lemmas \ref{Lemma:BlockingArea} and \ref{Lemma:BoundsOfCoverageGain}.
\begin{theorem}[Connectivity Probability of AAP]\label{Theorem:ConnectivityProb}\emph{The connectivity probability $p_\mathrm{c}$ is lower bounded by $p_\mathrm{c}^{(-)}$ as
\begin{align}\label{Eq:LowerBoundConnProb}
p_\mathrm{c}^{(-)} =  1 - \frac{\pi \lambda_\mathrm{b}\theta}{\Lambda_H^2} \int\limits_{L} \int\limits_{\Theta} \int\limits_{0}^{\Lambda_H^2} \mathcal{F}(r,\ell, \omega) r \mathrm{d} r  f_{\Theta}(\omega) \mathrm{d} \omega f_{L}(\ell) \mathrm{d} \ell,
\end{align}
where
\begin{align}\label{Eq:FunctionF}
\mathcal{F}(r,\ell, \omega)
&= \frac{\theta}{2} \l[ \Lambda_H^2 - \mathbf{1}\l( H_{\textrm{a}} > H_{\textrm{b}} \r)\l[\Lambda_H^2 - \l( d_\mathrm{L} + \Omega_H \r)^2\r]^{+} \r] - \frac{1}{2}d_\mathrm{L}^2\sin\theta,
\end{align}
$\Omega_H = \l({1 -  \frac{H_{\mathrm{b}} - H_{\mathrm{u}}}{H_{\mathrm{a}} - H_{\mathrm{u}}}}\r)^{-1}$, and  $\Lambda_H^2$ is specified in Lemma~\ref{Lemma:BlockingArea}.}
\end{theorem}
\begin{proof}
See Appendix \ref{Proof:theorem1}.
\end{proof}
\begin{remark}\label{Remark:TightnessOfBound}
\emph{The lower bound $p_\mathrm{c}^{(-)}$ becomes tighter when density of buildings, i.e., $\lambda_\mathrm{b}$, becomes smaller. This is because sparsely deployed buildings result in less spatial correlation.}
\end{remark}
\begin{remark}\label{Remark:OptimalConnProb}
\emph{Based on the discussion in Remark~\ref{Remark:OptimalAAPAltitude} and expression of $\mathcal{F}(r,\ell, \omega)$, the connectivity probability $p_\mathrm{c}$ can also be maximizing by optimizing the APP's altitude $H_\mathrm{u}$.}
\end{remark}

\section{Simulation Results}
In this section, we validate the analytical results via Monte Carlo simulation.  The radius of maximal coverage sphere is $R_{\mathrm{max}} = 100\mathrm{m}$. The height of building is $H_\mathrm{b} = 30\mathrm{m}$ and that of user is $H_{\textrm{u}} = 2\mathrm{m}$. The density of buildings is $\lambda_\mathrm{b} = 2\times 10^{-4} \mathrm{m}^2$. The length $\ell$ and orientation $\omega$ of building's line-segments follow  independently
and uniformly distributions. Specifically, $\ell$ is uniformly distributed in $(0, 15\mathrm{m}]$ and $\omega$ is uniformly distributed in $(0, \pi]$.

\begin{figure}[t]
\centering
\includegraphics[width=8cm]{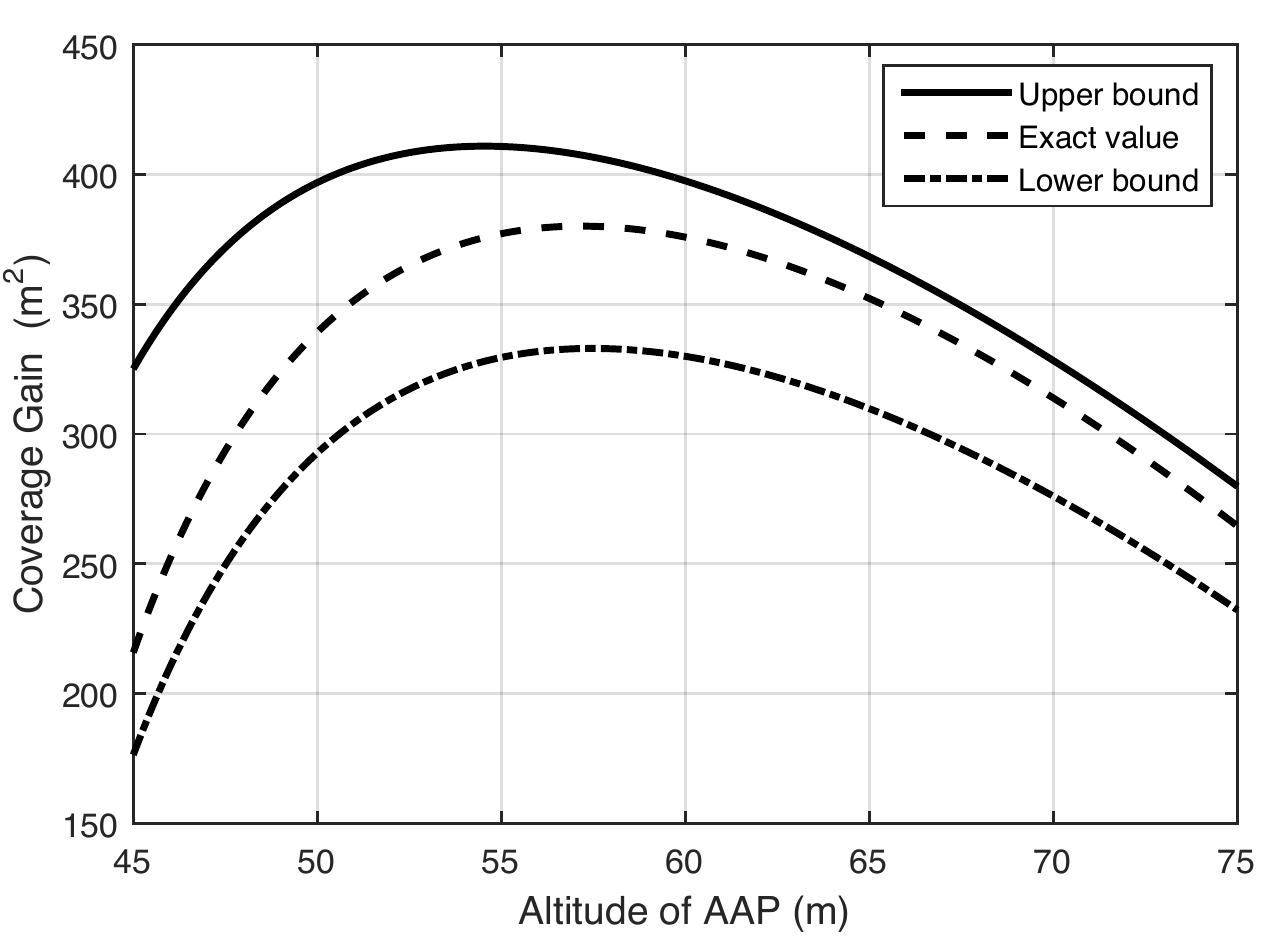}
\caption{The effect of the AAP's altitude on coverage gain $\mathcal{S}_{\mathrm{gain}}$. The coverage gain is shown to be a concave function of AAP's altitude. The parameters of building are set as $\{d_x, \ell, \omega\} = \{25\mathrm{m}, 6\mathrm{m}, \pi/4\}$. The upper and lower bounds are plotted based on \eqref{Eq:BoundsOfGain}. It is observed that $\mathcal{S}_{\mathrm{gain}}$ is well bounded and the lower bound becomes tighter when AAP's altitude is small and upper bound becomes tighter when AAP's altitude is large.}\label{CoverageGain}
\end{figure}

\begin{figure}[t]
\centering
\includegraphics[width=8cm]{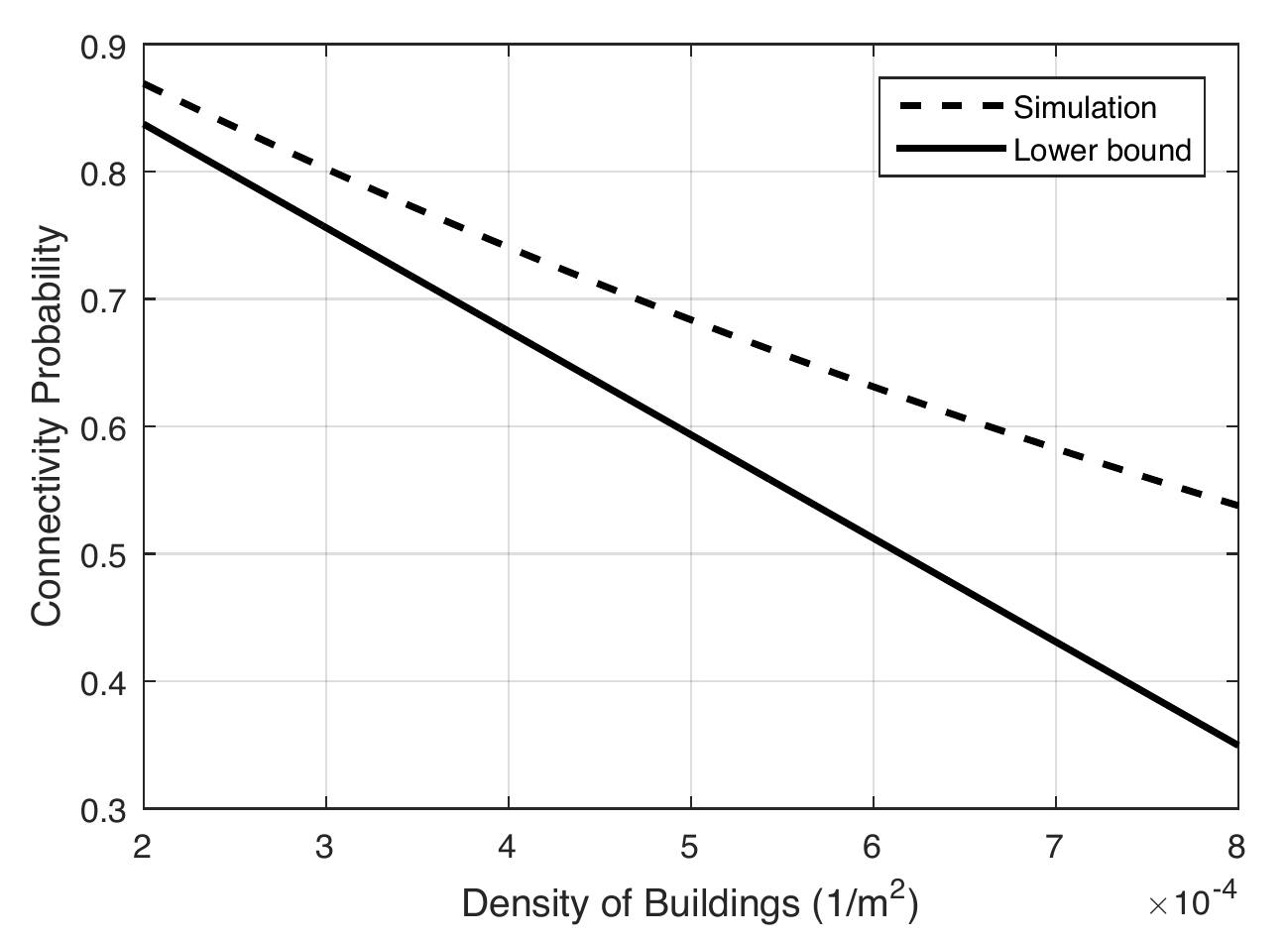}
\caption{The effect of building density on connectivity probability $p_\mathrm{c}$. The exact value of $p_\mathrm{c}$ is plotted via Monte Carlo simulation. The lower bound $p_\mathrm{c}^{(-)}$ is plotted based on \eqref{Eq:LowerBoundConnProb}. It is observed that connectivity probability decreases with building density and the lowered bound becomes tighter when building density is small.}\label{ConnectivityProb}
\end{figure}

Fig.~\ref{CoverageGain} shows the  coverage gain $\mathcal{S}_{\mathrm{gain}}$ calculated via \eqref{Eq:CoverageGain} and its bounds $\mathcal{S}_{\mathrm{gain}}^{(+)}$, $\mathcal{S}_{\mathrm{gain}}^{(-)}$ calculated via Lemma~\ref{Lemma:BoundsOfCoverageGain} versus the altitude of AAP $H_\mathrm{a}$. It is observed that $\mathcal{S}_{\mathrm{gain}}$ is well bounded by $\mathcal{S}_{\mathrm{gain}}^{(+)}$ and $\mathcal{S}_{\mathrm{gain}}^{(-)}$.  The lower bound becomes tighter when $H_\mathrm{a}$ is small and the upper bound becomes tighter when $H_\mathrm{a}$ is large. This agrees with the intuition because larger or smaller altitude of AAP results in larger or smaller coverage gain, respectively, which makes the bound tighter. Moreover, $\mathcal{S}_{\mathrm{gain}}$ is maximized via optimizing $H_\mathrm{a}$, which confirms the discussion given in Remark~\ref{Remark:OptimalAAPAltitude}.

In Fig.~\ref{ConnectivityProb}, we validate the lower bound of connectivity probability $p_\mathrm{c}$, i.e., $p_\mathrm{c}^{(-)}$, given in Theorem \ref{Theorem:ConnectivityProb} by comparing it with the exact value via Monte Carlo simulation. It is observed that both $p_\mathrm{c}$ and $p_\mathrm{c}^{(-)}$ decrease with building density $\lambda_\mathrm{b}$. More importantly, $p_\mathrm{c}^{(-)}$ becomes tighter when buildings are sparsely deployed, which aligns with the discussion in Remark~\ref{Remark:TightnessOfBound}.

\section{Conclusions and Future Work}
In this letter, we propose an analytical framework to define and characterize the connectivity in a mmWave A2X network. Based on the blockage model that buildings are modeled by the Boolean line-segment process with fixed height, we calculate the blockage area due to an arbitrary building and the connectivity probability of an AAP. Moreover, the AAP's altitude can be optimized to maximize the coverage area as well as the network connectivity. Future work will focus on studying the effects of spatial correlation of buildings on network connectivity and modelling a A2X network including multi-AAP's connections.

\section{Acknowledgments}
This work was supported in part by the National Science Foundation under Grant No. ECCS-1711702 and Hong Kong Research Grants Council under the Grants 17209917 and 17259416.

\appendix
\subsection{Proof of Theorem \ref{Theorem:ConnectivityProb}}\label{Proof:theorem1}
By omitting the spatial correlations between $\{\mathcal{S}_{\textrm{b}}(x)\}$, connectivity probability $p_\mathrm{c}$ defined in \eqref{Def:CoverageProb} is lowered bounded as follows.
\begin{align}\label{Eq:proofLowerBoundPC}
p_\mathrm{c} &\geq 1 - \E\l(\frac{\l|\sum_{x\in \{\Phi\cap\mathcal{O}(\Lambda_H)\}}\mathcal{S}_{\textrm{b}}(x)\r|}{\l|\mathcal{O}(\Lambda_H)\r|}\r) \nn \\
&\overset{(a)}{=} 1 - \frac{\pi \lambda_\mathrm{b}\theta}{\Lambda_H^2} \int\limits_{L} \int\limits_{\Theta} \int\limits_{0}^{\Lambda_H^2} \mathcal{S}_{\textrm{b}}(r) r \mathrm{d} r  f_{\Theta}(\omega) \mathrm{d} \omega f_{L}(\ell) \mathrm{d} \ell.
\end{align}
where $(a)$ follows the Campbell's theorem and ramdom shape theory \cite{gupta2018macrodiversity}. Based on \eqref{Eq:BlockingArea}, the blocking area $\mathcal{S}_{\textrm{b}}(r)$ can be upper bounded by $\frac{1}{2}\l[\theta \Lambda_H^2  - d_{\textrm{S}}^2\sin\theta\r] - \mathbf{1}\l( H_{\textrm{a}} > H_{\textrm{b}} \r)\mathcal{S}^{(-)}_{\textrm{gain}}(r)$.
Substituting the result above into \eqref{Eq:proofLowerBoundPC} gives the lower bound of $p_\mathrm{c}$, say $p_\mathrm{c}^{(-)}$, in Theorem \ref{Theorem:ConnectivityProb}.

\bibliographystyle{ieeetr}

\end{document}